\documentclass[twocolumn,amsmath,amssymb,10pt,aps]{revtex4}
  
\usepackage[utf8]{inputenc}
\usepackage[T1]{fontenc}

\usepackage{amsmath}
\usepackage{amsfonts}
\usepackage{graphicx}
\usepackage{yfonts}
\usepackage{color}
\usepackage[normalem]{ulem}
\usepackage{amsthm}
\usepackage{bm}
\usepackage{bbm}
\usepackage{mathtools}
\usepackage{array}
\usepackage{placeins}
\usepackage{enumitem}

\newcommand{\der}{\,\mathrm{d}}

\def\<{\langle}
\def\>{\rangle}

\def\oper{{\mathchoice{\rm 1\mskip-4mu l}{\rm 1\mskip-4mu l}
{\rm 1\mskip-4.5mu l}{\rm 1\mskip-5mu l}}}
\DeclareMathAlphabet\mathbfcal{OMS}{cmsy}{b}{n}

\mathchardef\mhyphen="2D 

\newtheorem{Remark}{Remark}
\newtheorem{Proposition}{Proposition}
\newtheorem{Example}{Example}

\begin{document}

\title{Markovian semigroup from mixing non-invertible dynamical maps}

\author{Katarzyna Siudzi{\'n}ska}
\affiliation{Institute of Physics, Faculty of Physics, Astronomy and Informatics \\  Nicolaus Copernicus University, ul. Grudzi\k{a}dzka 5/7, 87--100 Toru{\'n}, Poland}

\begin{abstract}
We analyze the convex combinations of non-invertible generalized Pauli dynamical maps. By manipulating the mixing parameters, one can produce a channel with shifted singularities, additional singularities, or even no singularities whatsoever. In particular, we show how to use non-invertible dynamical maps to produce the Markovian semigroup. Interestingly, the maps whose mixing results in a semigroup are generated by the time-local generators and time-homogeneous memory kernels that are not regular; i.e., their formulas contain infinities. Finally, we show how the generators and memory kernels change after mixing the corresponding dynamical maps.
\end{abstract}

\flushbottom

\maketitle

\thispagestyle{empty}

\section{Introduction}

It is impossible to completely isolate a physical system from any external interactions. Hence, to properly describe its dynamics, it becomes necessary to apply the methods used in the theory of open quantum systems \cite{BreuerPetr,Weiss}. The evolution $\rho\mapsto\rho^\prime=\Lambda[\rho]$ of an open quantum system in an initial state $\rho$ is characterized by a completely positive, trace-preserving map, also known as the {\it quantum channel}. If the system-environment interactions allow for a time-continuous description of the system dynamics, then one can treat time $t>0$ as a parameter. The time-parameterized family $\Lambda(t)$ of quantum channels, together with the initial condition $\Lambda(0)=\oper$, is called the {\it quantum dynamical map}. In the case of weak coupling between the system and its environment, $\Lambda(t)$ is the Markovian semigroup \cite{BreuerPetr}. This means that it solves the master equation
\begin{equation}\label{SGE}
\dot{\Lambda}(t)=\mathcal{L}\Lambda(t)
\end{equation}
with the Gorini-Kossakowski-Sudarshan-Lindblad (GKSL) generator \cite{GKS,L}
\begin{equation}
\mathcal{L}[\rho]=-\frac i\hbar [H,\rho]+\sum_\alpha \gamma_\alpha
\left(V_\alpha\rho V_\alpha^\dagger-\frac 12\{V_\alpha^\dagger V_\alpha,\rho\}\right),
\end{equation}
where $H$ is the effective Hamiltonian, $V_\alpha$ denote the noise operators, and $\gamma_\alpha\geq 0$ are the decoherence rates.
There are two ways to generalize eq. (\ref{SGE}), so that it also describes physical systems with memory. First, one can replace the time-independent $\mathcal{L}$ with the time-local generator $\mathcal{L}(t)$ whose decoherence rates $\gamma_\alpha(t)$ no longer need to be positive \cite{Andersson}. Second, one can replace the right hand-side of the master equation with the time integral \cite{Nakajima,Zwanzig}, so that
\begin{equation}
\dot{\Lambda}(t)=\int_0^tK(t,\tau)\Lambda(\tau)\der\tau,
\end{equation}
where $K(t,\tau)$ is the memory kernel. For a given dynamical map, one can usually find both the time-local generator and the memory kernel \cite{ENM}. However, determining the necessary and sufficient conditions for $\mathcal{L}(t)$ and $K(t,s)$ to generate a legitimate (completely positive and trace-preserving) dynamical map is still an open problem.

An increasing attention is being given to mixing quantum dynamical maps. Its rapid development began after it was proven that eternally non-Markovian qubit evolution can be obtained as a convex combination of two Pauli dynamical semigroups \cite{ENM,Nina}. This result was soon extended to qudits. It was shown that certain mixtures of generalized Pauli dynamical semigroups result in the evolution with one \cite{mub_final} or more \cite{ICQC} eternally negative decoherence rates. The fact that non-Markovianity emerges from mixtures of Markovian semigroups appeared to be non-intuitive at first. However, this behavior was later explained within a characterization of non-Markovianity in terms of the information flow \cite{Amato}. Interestingly, it is also possible to obtain a Markovian semigroup by taking a combination of two non-Markovian qudit evolutions \cite{CCnM}. Recently, more general mixtures of Pauli dynamical semigroups have been analyzed \cite{Jagadish2}. Some authors went beyond the Markovian semigroups and considered combinations of CP-divisible dynamical maps \cite{CCMS,Jagadish3}. However, the maps they were mixing were always invertible; that is, $\Lambda^{-1}(t)$ always existed and was well-defined.

In this paper, we go a step further and analyze convex combinations of non-invertible dynamical maps. We consider the qudit evolution provided by the generalized Pauli channels. One observes an interesting non-intuitive behavior: mixing non-invertible maps changes the singular points of the resulting map. One can shift the initial singularities to different points in time. Other singularities can be produced in addition to the existing ones. One can also remove all singularities and produce an invertible dynamical map, even a Markovian semigroup.

In the following sections, we present two ways to generate mixtures of non-invertible dynamical maps. For the master equations with time-local generators, the correspondence between the input and output decoherence rates is non-linear. Moreover, if the dynamical map is non-invertible, then the associated generator is singular. In the memory kernel approach, we find the class that corresponds to the convex combinations of generalized Pauli channels. In the Conclusions, we list the open questions that arose during our research.

\section{Mixing generalized Pauli channels}

Mixed unitary evolution of a qubit is described by the Pauli channels
\begin{equation}\label{Pauli}
\Lambda[X]=\sum_{\alpha=0}^3p_\alpha\sigma_\alpha X\sigma_\alpha,
\end{equation}
where $p_\alpha$ is a probability distribution and $\sigma_\alpha$ the Pauli matrices. Alternatively, it can be defined through the eigenvalue equations $\Lambda[\sigma_\alpha]=\lambda_\alpha\sigma_\alpha$, where $\lambda_0=1$ and
\begin{equation}
\lambda_\alpha=2(p_0+p_\alpha)-1
\end{equation}
for $\alpha=1,2,3$. These channels possess many important properties that make them relatively easy to analyze. Namely, the Pauli channels have eigenvectors that do not depend on the choice of $p_\alpha$. They preserve Hermiticity, from which it follows that their eigenvalues $\lambda_\alpha$ are real. Moreover, $\lambda_0=1$, which means these channels are unital.

A generalization of the Pauli channels to the qudit evolution that keeps their characteristic features was introduced by Nathanson and Ruskai \cite{Ruskai}. In their construction, the authors used the collections of $N$ orthonormal bases $\{\psi_k^{(\alpha)},k=0,\ldots,d-1\}$ in $\mathbb{C}^d$ for which $|\<\psi_k^{(\alpha)}|\psi_l^{(\beta)}\>|^2=1/d$ whenever $\alpha\neq\beta$. Such bases are known as {\it mutually unbiased bases} (MUBs). Assume that the Hilbert space $\mathcal{H}\simeq\mathbb{C}^d$ admits the maximal number of $N=d+1$ MUBs numbered by $\alpha=1,\ldots,d+1$. Then, one can use the rank-1 projectors $P_k^{(\alpha)}=|\psi_k^{(\alpha)}\>\<\psi_k^{(\alpha)}|$ to introduce $d+1$ unitary operators
\begin{equation}
U_\alpha=\sum_{l=0}^{d-1}\omega^lP_l^{(\alpha)},\qquad \omega=e^{2\pi i/d}.
\end{equation}
The generalized Pauli channels are defined by \cite{Ruskai,mub_final}
\begin{equation}
\Lambda=p_0\oper+\frac{1}{d-1}\sum_{\alpha=1}^{d+1}p_\alpha\mathbb{U}_\alpha,
\end{equation}
where $p_\alpha$ is a probability distribution and
\begin{equation}
\mathbb{U}_\alpha[\rho]=\sum_{k=1}^{d-1}U_\alpha\rho U_\alpha^\dagger.
\end{equation}
For $d=2$, one recovers the Pauli channels from eq. (\ref{Pauli}).
Throughout this paper, we characterize the generalized Pauli channels by their eigenvalues $\lambda_\alpha$, where $\Lambda[U_\alpha^k]=\lambda_\alpha U_\alpha^k$. They are related to $p_0$ and $p_\alpha$ via
\begin{equation}
\lambda_\alpha=\frac{1}{d-1}\left[d(p_\alpha+p_0)-1\right].
\end{equation}
Recall that $\Lambda$ is invertible if and only if there exists $\Lambda^{-1}$ such that $\Lambda\Lambda^{-1}=\Lambda^{-1}\Lambda=\oper$. In terms of its eigenvalues, it means that $\Lambda$ is invertible as long as $\lambda_\alpha\neq 0$ for all $\alpha$. Now, the complete positivity criteria $p_\alpha\geq 0$, $\sum_{\alpha=0}^{d+1}p_\alpha=1$ translate to the generalized Fujiwara-Algoet conditions \cite{Fujiwara,Ruskai}
\begin{equation}
-\frac{1}{d-1}\leq\sum_{\beta=1}^{d+1}\lambda_\beta\leq 1+d\min_\alpha\lambda_\alpha.
\end{equation}

In the theory of open quantum systems, the evolution of a quantum system is described by a family of time-parameterized channels $\Lambda(t)$, $t\geq 0$, with the initial condition $\Lambda(0)=\oper$. Observe that the generalized Pauli dynamical map
\begin{equation}
\Lambda(t)=p_0(t)\oper+\frac{1}{d-1}\sum_{\alpha=1}^{d+1}p_\alpha(t)\mathbb{U}_\alpha
\end{equation}
can be obtained by mixing the dynamical maps
\begin{equation}
\Lambda_\alpha(t)=(1-\pi_\alpha(t))\oper+\frac{\pi_\alpha(t)}{d-1}\mathbb{U}_\alpha
\end{equation}
with $0\leq \pi_\alpha(t)\leq 1$. Indeed, one has
\begin{equation}
\Lambda(t)=\frac{1}{d+1}\sum_{\alpha=1}^{d+1}\Lambda_\alpha(t),
\end{equation}
where $p_\alpha(t)=\frac{1}{d+1}\pi_\alpha(t)$. In other words, any generalized Pauli dynamical map is a mixture of $\Lambda_\alpha(t)$ with identical mixing parameters $x_\alpha=\frac{1}{d+1}$ and arbitrary $\pi_\alpha(t)$. In this paper, we consider the mixtures of $\Lambda_\alpha(t)$ with arbitrary mixing parameters $x_\alpha\geq 0$, $\sum_{\alpha=1}^{d+1}x_\alpha=1$, and identical $\pi_\alpha(t)\equiv p(t)$, so that
\begin{equation}\label{Lt}
\Lambda(t)=\sum_{\alpha=1}^{d+1}x_\alpha\Lambda_\alpha(t)=(1-p(t))\oper+\frac{p(t)}{d-1}
\sum_{\alpha=1}^{d+1} x_\alpha\mathbb{U}_\alpha.
\end{equation}
The eigenvalues of such $\Lambda(t)$ read
\begin{equation}\label{lat}
\lambda_\alpha(t)=1-\frac{d}{d-1}(1-x_\alpha)p(t)=x_\alpha+(1-x_\alpha)\lambda(t),
\end{equation}
where $\lambda(t)\in[-\frac{1}{d-1},1]$ is a $d(d-1)$-times degenerated eigenvalue of $\Lambda_\alpha(t)$ to $U_\beta^k$, $\beta\neq\alpha$ (the eigenvalue to $U_\alpha^k$ is $1$).
Note that $\lambda_\alpha(t)$ cannot be negative for a positive $\lambda(t)$. Therefore, non-invertible dynamical maps cannot be constructed by mixing invertible generalized Pauli dynamical maps. The converse is not true, as mixtures of $d+1$ non-invertible generalized Pauli dynamical maps can produce invertible channels.

\begin{Proposition}
A mixture of non-invertible $\Lambda_\alpha(t)$ is invertible if and only if $x_\alpha\neq 0$ for all $\alpha=1,\ldots,d+1$ and
\begin{equation}
\lambda(t)>-\frac{x_{\min}}{1-x_{\min}},\qquad x_{\min}=\min_\alpha x_\alpha.
\end{equation}
\end{Proposition}

\begin{proof}
First, note that if $x_\alpha=0$ for a fixed $\alpha$, then $\lambda_\alpha(t)=\lambda(t)$ by eq. (\ref{lat}). By assumption, $\Lambda_\alpha(t)$ are non-invertible, which means that $\lambda(t)=0$ for some $t>0$. Therefore, if even a single $x_\alpha=0$, the resulting map $\Lambda(t)$ cannot be inverted. The condition $x_\alpha\neq 0$ for all $\alpha$ follows.

Now, if all $x_\alpha\neq 0$, the necessary and sufficient condition for the invertibility of $\Lambda(t)$ is the positivity of its eigenvalues,
\begin{equation}
\lambda_\alpha(t)=x_\alpha+(1-x_\alpha)\lambda(t)>0.
\end{equation}
The above inequality holds if and only if
\begin{equation}\label{con}
\lambda(t)>-\frac{x_\alpha}{1-x_\alpha}
\end{equation}
for all $\alpha=1,\ldots,d+1$. It is enough to check this condition for the smallest value of $x_\alpha$, which ends the proof.
\end{proof}

\begin{Example}
Let us take a combination of $d+1$ dynamical maps $\Lambda_\alpha(t)$ with identical mixing parameters $x_\alpha=\frac{1}{d+1}$. The eigenvalues of the resulting family of channels are given by
\begin{equation}\label{la}
\lambda_\alpha(t)=\frac{1+d\lambda(t)}{d+1}.
\end{equation}
According to condition (\ref{con}), they are positive if and only if $\lambda(t)>-1/d$. Hence, one can have $\lambda(t_\ast)=0$ for $t_\ast>0$, and the map $\Lambda(t_\ast)$ is still invertible.
\end{Example}

As a special case of the evolution provided by eq. (\ref{Lt}), consider the mixture with $k\leq d+1$ identical non-zero $x_\alpha=1/k$. Then, one has
\begin{equation}\label{lak}
\lambda_\alpha(t)=\left\{
\begin{aligned}
&\frac{1+(k-1)\lambda(t)}{k}, &1\leq\alpha\leq k,\\
&\lambda(t), &k+1\leq\alpha\leq d+1.
\end{aligned}\right.
\end{equation}
Observe that if $k\leq d$, then there exist dynamical maps $\Lambda(t)$ with two distinct $\lambda_\alpha(t)$ that reach zero at some point.
Assume that $\lambda(t_\ast)=0$ for a time $t_\ast>0$ and $\lambda(t_\#)=-\frac{1}{k-1}$ for some $t_\#>t_\ast$. Then, according to eq. (\ref{lak}),
\begin{equation}
\begin{split}
\lambda_\alpha(t_\#)=0,&\qquad 1\leq \alpha\leq k,\\
\lambda_\alpha(t_\ast)=0,&\qquad k+1\leq\alpha\leq d+1.
\end{split}
\end{equation}
Recall that $-\frac{1}{d-1}\leq\lambda(t)\leq 1$, and hence it is possible to have $\lambda(t)=-\frac{1}{k-1}$ only for $k=d$. Notably, for $d=2$, this is the entire admissible range of $k$.

\begin{Example}\label{ExZo}
Construct the convex combination of $d$ generalized Pauli dynamical maps characterized by
\begin{equation}\label{Zo}
\lambda(t)=e^{-Zt}\cos\omega t,\qquad Z,\omega\geq 0.
\end{equation}
The roots of $\lambda(t)$ are $t_\ast(N)=\frac{\pi}{2\omega}+N\pi$, $N\in\mathbb{Z}_+$. The global minimum corresponds to $t_\#=\frac{\pi}{\omega}$ and is equal to $\lambda(t_\#)=-\frac{1}{d-1}$ provided that there is the following relation between $Z$ and $\omega$,
\begin{equation}
Z=\frac{\omega}{\pi}\ln(d-1).
\end{equation}
Indeed, $\Lambda(t)$ has more singular points than $\Lambda_\alpha(t)$, as for $1\leq\alpha\leq d$,
\begin{equation}
\lambda_\alpha(t)=\frac{1+(d-1)e^{-Zt}\cos\omega t}{d}
\end{equation}
adds its own singular point at $t_\#$.
\end{Example}

Through mixtures, one can not only add and remove the singularities of dynamical maps but also shift them to another point in time.

\begin{Proposition}
A mixture $\Lambda(t)$ of dynamical maps $\Lambda_\alpha(t)$ with $N$ singularities at $t_\ast(N)>0$ has all its singularities shifted to $t_\#(N)$ if and only if $k=d+1$ and $\lambda(t)$ satisfies
\begin{equation}
\lambda(t_\ast(N))=0,\qquad\lambda(t_\#(N))=-\frac 1d.
\end{equation}
\end{Proposition}

Note that the first singularity point $t_\ast$ can be only shifted to a later time $t_\#>t_\ast$.

\begin{Example}\label{Exo}
Let us take a convex combination of three Pauli dynamical maps $\Lambda_\alpha(t)$ (for $d=2$) with $x_\alpha=1/3$ amd $\lambda(t)=\cos\omega t$, $\omega>0$. The eigenvalues of the resulting map are
\begin{equation}\label{o}
\lambda_\alpha(t)=\frac{1+2\cos\omega t}{3}.
\end{equation}
Their first root corresponds to
\begin{equation}
t_\#(0)=\frac{1}{\omega}\arccos\left(-\frac 12\right),
\end{equation}
which is $t_\ast(0)=\frac{\pi}{2\omega}$ shifted by $\Delta t=t_\#(0)-t_\ast(0)$. Analogical relations can be seen between the following roots, $t_\ast(N)=\frac{\pi}{2\omega}+N\pi$ and
\begin{equation}
t_\#(N)=\pm\frac{1}{\omega}\left[\arccos\left(-\frac 12\right)+\pi N\right],
\end{equation}
for $N\geq 1$, where the shift is $\pm\Delta t$ (see Fig.\ref{shift}).
\end{Example}

\begin{figure}[htb!]
  \includegraphics[width=0.45\textwidth]{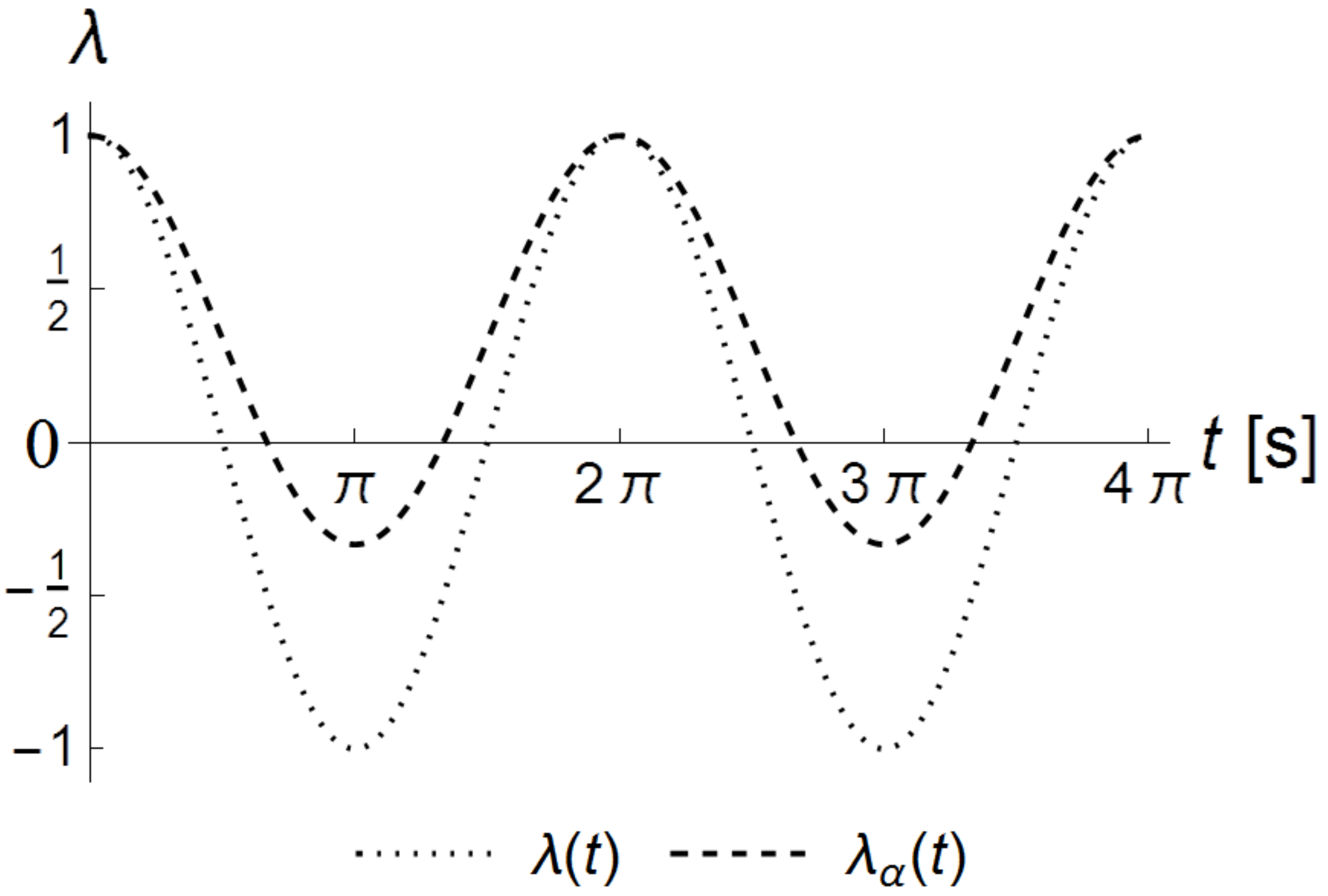}
\caption{
The eigenvalues of $\Lambda_\alpha(t)$ and their mixture $\Lambda(t)$ from Example \ref{Exo} for $\omega=1/\mathrm{s}$.
}
\label{shift}
\end{figure}

\section{Markovian semigroups and time-local generators}

Dynamical maps are usually constructed as the solutions to the master equations, which are the evolution equations for open quantum systems. The simplest generalized Pauli dynamical map is the Markovian semigroup $\Lambda(t)=e^{t\mathcal{L}}$ with the GKSL generator \cite{mub_final}
\begin{equation}
\mathcal{L}=\sum_{\alpha=1}^{d+1}\gamma_\alpha\mathcal{L}_\alpha,
\end{equation}
where
\begin{equation}\label{La}
\mathcal{L}_\alpha=\frac 1d \left[\mathbb{U}_\alpha-(d-1)\oper\right],
\end{equation}
and the decoherence rates $\gamma_\alpha\geq 0$. Recall that the dynamical semigroup describes the quantum systems that are weakly coupled to the environment. In the presence of strong coupling, it becomes necessary to include non-Markovian memory effects. One way to do this is to generalize the semigroup master equation to
\begin{equation}
\dot{\Lambda}(t)=\mathcal{L}(t)\Lambda(t)
\end{equation}
with the time-local generator 
\begin{equation}
\mathcal{L}(t)=\sum_{\alpha=1}^{d+1}\gamma_\alpha(t)\mathcal{L}_\alpha,
\end{equation}
whose decoherence rates no longer have to be positive. The relation between $\lambda_\alpha(t)$ and $\gamma_\alpha(t)$ is \cite{mub_final}
\begin{equation}
\lambda_\alpha(t)=\exp[\Gamma_\alpha(t)-\Gamma_0(t)],
\end{equation}
where $\Gamma_\alpha(t)=\int_0^t\gamma_\alpha(\tau)\der\tau$ and $\gamma_0(t)=\sum_{\alpha=1}^{d+1}\gamma_\alpha(t)$.

Now, observe that $\Lambda_\alpha(t)$ can be generated via a time-local generator $\gamma(t)\mathcal{L}_\alpha$ with
\begin{equation}
\gamma(t)=-\frac{\dot{\lambda}(t)}{\lambda(t)}.
\end{equation}
Hence, whenever $\lambda(t)$ is singular, the corresponding rate $\gamma(t)$ goes to infinity, and the generator is not regular. One finds the relation between $\gamma(t)$ and $\gamma_\alpha(t)$,
\begin{equation}\label{gammas}
\gamma_\alpha(t)=-\frac{\gamma(t)(1-x_\alpha)}{1+(e^{\Gamma(t)}-1)x_\alpha}+\gamma_0(t),
\end{equation}
where $\Gamma(t)=\int_0^t\gamma(\tau)\der\tau$ and
\begin{equation}
\gamma_0(t)=\frac 1d \sum_{\alpha=1}^{d+1}
\frac{\gamma(t)(1-x_\alpha)}{1+(e^{\Gamma(t)}-1)x_\alpha}.
\end{equation}
Therefore, $\gamma_\alpha(t)$ depends on $\gamma(t)$ and its integral in a non-linear way.
For $x_\alpha=1/k$, one has
\begin{equation}\label{g1}
\gamma_\alpha(t)=
\frac{\gamma(t)}{d}\frac{d-(k-1)(1-e^{-\Gamma(t)})}{1+(k-1)e^{-\Gamma(t)}}
\end{equation}
for $1\leq\alpha\leq k$ and
\begin{equation}
\gamma_\alpha(t)=
-\frac{\gamma(t)}{d}\frac{(k-1)(1-e^{-\Gamma(t)})}{1+(k-1)e^{-\Gamma(t)}}
\end{equation}
for $k+1\leq\alpha\leq d+1$.
From the complete positivity conditions, one has $\Gamma(t)\geq 0$. Hence, whenever $\gamma(t)\geq 0$, the rates $\gamma_\alpha(t)$ are positive for $1\leq\alpha\leq k$ and negative for $k+1\leq\alpha\leq d+1$. For $\gamma(t)\leq 0$, the converse is true. Observe that when $\gamma(t)\to\infty$, the generator $\mathcal{L}(t)$ of the mixture is not regular for $2\leq k\leq d$, as its decoherence rates $\gamma_\alpha(t)\to\pm\infty$.

\begin{Proposition}
For $k=d+1$, there exist mixtures $\Lambda(t)$ with generators $\mathcal{L}(t)$ that are regular for all $t\geq 0$.
\end{Proposition}

\begin{proof}
If $k=d+1$, then the rates in eq. (\ref{g1}) can be rewritten into
\begin{equation}\label{ga}
\gamma_\alpha(t)=\frac{\gamma(t)}{d+e^{\Gamma(t)}},
\end{equation}
so for $\gamma(t)\to\infty$ they produce an indeterminate form. However, they can be always regular, as is easily seen if one takes
\begin{equation}\label{lam2}
\lambda(t)=\frac{(d+1)e^{-R(t)}-1}{d},
\end{equation}
where $R(t)=\int_0^tr(\tau)\der\tau$ and $r(t)<\infty$ for finite times $t\geq 0$. Then, the associated maps $\Lambda_\alpha(t)$ are generated via $\gamma(t)\mathcal{L}_\alpha(t)$ with
\begin{equation}
\gamma(t)=-\frac{\dot{\lambda}(t)}{\lambda(t)}=\frac{r(t)(d+1)}{d+1-e^{R(t)}}.
\end{equation}
From eq. (\ref{ga}), the convex combination $\Lambda(t)$ has the corresponding
\begin{equation}
\gamma_\alpha(t)=\frac{r(t)(d+1)}{d+1-e^{R(t)}}\left[d+
\exp\left(\int_0^t\frac{(d+1)r(\tau)\der\tau}{d+1-e^{R(\tau)}}\right)\right]^{-1},
\end{equation}
which is exactly $\gamma_\alpha(t)=r(t)/d$ due to
\begin{equation}
\int_0^t\frac{(d+1)r(\tau)\der\tau}{d+1-e^{R(\tau)}}=\ln\frac{de^{R(t)}}{d+1-e^{R(t)}}.
\end{equation}
Hence, $\mathcal{L}(t)$ is a regular time-local generator.
\end{proof}

\begin{Remark}\label{MSG}
From the above proposition, it follows that a Markovian semigroup can result from mixing non-invertible dynamical maps. For $r(t)=r$, one has
\begin{equation}\label{lam}
\lambda(t)=\frac{(d+1)e^{-rt}-1}{d},
\end{equation}
which is evidently non-invertible, as $\lambda(t_\ast)=0$ for $t_\ast=\ln(d+1)/r$. The eigenvalues of the resulting map $\Lambda(t)$ are $\lambda_\alpha(t)=e^{-rt}$, which means that it is the Markovian semigroup generated via $\mathcal{L}$ with $\gamma_\alpha=r/d$.
\end{Remark}

\section{Memory kernels}


Now, let us include the memory effects of quantum dynamics through the master equations with non-local memory kernels. The memory kernel master equations are integro-differential, hence the state of the system depends on its entire history (all previous states). Consider the evolution provided by the Nakajima-Zwanzig equation \cite{Nakajima,Zwanzig}
\begin{equation}
\dot{\Lambda}(t)=\int_0^tK(t-\tau)\Lambda(\tau)\der\tau,
\end{equation}
where $K(t)$ is a time-homogeneous memory kernel.
For the generalized Pauli channels, one takes
\begin{equation}\label{K}
K(t)=\sum_{\alpha=1}^{d+1}k_\alpha(t)\mathcal{L}_\alpha,
\end{equation}
where $\mathcal{L}_\alpha$ are given by eq. (\ref{La}). The Markovian semigroup follows for $K(t)=\delta(t)\mathcal{L}$.
Due to the eigenvectors $U_\alpha^k$ of $K(t)$ being time-independent, it is enough to solve the master equations for the eigenvalues
\begin{equation}
\dot{\lambda}_\alpha(t)=\int_0^t\kappa_\alpha(t-\tau)\lambda_\alpha(\tau)\der\tau,
\end{equation}
where $K(t)[U_\alpha^k]=\kappa_\alpha(t)U_\alpha^k$. The relation between $\kappa_\alpha(t)$ and $k_\alpha(t)$ reads
\begin{equation}
\kappa_\alpha(t)=k_\alpha(t)-k_0(t)
\end{equation}
with $k_0(t)=\sum_{\alpha=1}^{d+1}k_\alpha(t)$. Now, in the Laplace transform domain, the general solution of the Nakajima-Zwanzig equation has the form
\begin{equation}\label{lambda_s}
\widetilde{\lambda}_\alpha(s)=\frac{1}{s-\widetilde{\kappa}_\alpha(s)},
\end{equation}
where $\widetilde{f}(s)=\int_0^\infty f(t)e^{-st} \der t$ is the Laplace transform of $f(t)$.

\begin{Example}
Let us take the Pauli dynamical map $\Lambda_\alpha(t)$ with $\lambda(t)=\cos\omega t$. It is the solution of the master equations $\dot{\Lambda}_\alpha(t)=\gamma(t)\mathcal{L}_\alpha\Lambda_\alpha(t)$ and $\dot{\lambda}(t)=\int_0^t\kappa(t-\tau)\lambda(\tau)\der\tau$ with
\begin{equation}
\gamma(t)=\omega\tan\omega t,\qquad \kappa(t)=-\omega^2,
\end{equation}
respectively. Note that the memory kernel is regular while the time-local generator is singular, as $\gamma(t)\to\infty$ for $t=\frac{\pi}{\omega}(N+1/2)$, $N\in\mathbb{Z}$.
\end{Example}

The complementary behavior of generators and memory kernels was first observed in \cite{local-vs-non-local}. It has been shown that when one is simple and regular, the other one is complex or even singular.
However, it turns out that the dynamical maps $\Lambda_\alpha(t)$ from Remark \ref{MSG}, which can be mixed into the Markovian semigroup, are produced by the generators and memory kernels that are both singular.

\begin{Example}
The generalized Pauli dynamical map $\Lambda_\alpha(t)$ with
\begin{equation}
\lambda(t)=\frac{(d+1)e^{-rt}-1}{d},\qquad r>0,
\end{equation}
can be generated by the time-local generator $\gamma(t)\mathcal{L}_\alpha$ or the memory kernel $-\kappa(t)\mathcal{L}_\alpha$, where
\begin{equation}
\gamma(t)=\frac{r(d+1)}{d+1-e^{rt}}
\end{equation}
and
\begin{equation}
\kappa(t)=-\frac{r(d+1)}{d}\left[\delta(t)+\frac{r}{d}e^{rt/d}\right].
\end{equation}
Observe that $\gamma(t)\to\infty$ for $t=\ln(d+1)/r$ and $\kappa(t)\to-\infty$ for $t\to\infty$. Hence, both the time-local generator and the memory kernel are not regular, even though the solution $\Lambda_\alpha(t)$ is legitimate.
\end{Example}


In \cite{memory}, the necessary and sufficient conditions for admissible memory kernels have been given. Namely, $K(t)$ defined in eq. (\ref{K}) is legitimate if and only if its eigenvalues in the Laplace transform domain are equal to
\begin{equation}\label{th_1}
\widetilde{\kappa}_\alpha(s)=-\frac{s\widetilde{\ell}_\alpha(s)}{1-\widetilde{\ell}_\alpha(s)},
\end{equation}
where $\ell_\alpha(t)$ satisfy
\begin{equation}\label{CON}
\begin{split}
\int_0^t \ell_\alpha(\tau)\der\tau&\geq 0,\\
\sum_{\beta=1}^{d+1}\int_0^t\ell_\beta(\tau)\der\tau&\leq\frac{d^2}{d-1},\\
\sum_{\beta=1}^{d+1}\int_0^t\ell_\beta(\tau)\der\tau&\geq d\,\int_0^t \ell_\alpha(\tau)\der\tau
\end{split}
\end{equation}
for all $\alpha=1,\ldots,d+1$. The corresponding solution is
\begin{equation}\label{}
\lambda_\alpha(t)=1-\int_0^t\ell_\alpha(\tau)\der\tau.
\end{equation}
In an analogical way, let us reparameterize $\lambda(t)$ as
\begin{equation}
\lambda(t)=1-\int_0^t\ell(\tau)\der\tau.
\end{equation}
The associated dynamical map is completely positive if and only if
\begin{equation}\label{ell}
0\leq\int_0^t\ell(\tau)\der\tau\leq\frac{d}{d-1},
\end{equation}
which is a direct consequence of the generalized Fujiwara-Algoet conditions $-\frac{1}{d-1}\leq\lambda(t)\leq 1$ for $\Lambda_\alpha(t)$. This map is generated via the memory kernel $-\kappa(t)\mathcal{L}_\alpha$, whose eigenvalues are $0$ (to the eigenvectors $U_\alpha^k$) and the $d(d-1)$-times degenerated
\begin{equation}\label{k}
\widetilde{\kappa}(s)=-\frac{s\widetilde{\ell}(t)}{1-\widetilde{\ell}(t)}.
\end{equation}
Now, $\kappa_\alpha(t)$ for the memory kernel generating the convex combination $\Lambda(t)$ are related to $\kappa(t)$ in the following way,
\begin{equation}\label{ka}
\widetilde{\kappa}_\alpha(s)=\frac{(1-x_\alpha)s\widetilde{\kappa}(s)}{s-x_\alpha\widetilde{\kappa}(s)}.
\end{equation}
Interestingly, the dependence between $\kappa_\alpha(t)$ and $\kappa(t)$ is less involved than between $\gamma_\alpha(t)$ and $\gamma(t)$ in eq. (\ref{gammas}).
Also, it can be shown that $\ell_\alpha(t)$ depend on $\ell(t)$ via
\begin{equation}\label{ela}
\ell_\alpha(t)=(1-x_\alpha)\ell(t),
\end{equation}
which is a special case of $\ell_\alpha(t)=\ell(t)/a_\alpha$ considered in \cite{chlopaki,memory}. Therefore, the mixture from eq. (\ref{Lt}) is generated via $K(t)$ with the eigenvalues
\begin{equation}\label{ka2}
\widetilde{\kappa}_\alpha(s)=-\frac{s(1-x_\alpha)\widetilde{\ell}(s)}
{1-(1-x_\alpha)\widetilde{\ell}(s)}.
\end{equation}
Note that condition (\ref{ell}) for $\ell(t)$ is also necessary and sufficient for the legitimacy of the above kernel.

\begin{Example}
If $d=2$ and $\ell(t)=\omega\sin\omega t$, the corresponding memory kernels are $\kappa(t)=-\omega^2$ and
\begin{equation}
\kappa_\alpha(t)=-\omega^2(1-x_\alpha)\cos\left[\sqrt{x_\alpha}\omega t\right].
\end{equation}
Hence, the mixture of dynamical maps generated by constant memory kernels solves the master equation with an oscillating $K(t)$. For $x_\alpha=1/3$, one recovers the evolution from Example \ref{Exo}.
\end{Example}

\begin{Example}
Let us take $\ell(t)=e^{-Zt}(Z\cos\omega t+\omega\sin\omega t)$ with $Z\geq \frac{\omega}{\pi}\ln(d-1)$.
This choice results in the memory kernels with the eigenvalues $\kappa(t)=-Z\delta(t)-\omega^2e^{-Zt}$ and
\begin{equation}
\begin{split}
\kappa_\alpha(t)=&-(1-x_\alpha)\Bigg[Z\delta(t)+\frac{1}{P_\alpha}e^{-Z(1+x_\alpha)t/2}\\
&\times
\left(-A_\alpha\sin\frac{P_\alpha t}{2}+B_\alpha\cos\frac{P_\alpha t}{2}\right)\Bigg],
\end{split}
\end{equation}
where
\begin{align*}
A_\alpha&=ZP_\alpha^2+(1-x_\alpha)Z(\omega^2+Z^2),\\
B_\alpha&=P_\alpha(\omega^2-x_\alpha Z^2),\\
P_\alpha&=\sqrt{4\omega^2x_\alpha-Z^2(1-x_\alpha)^2}.
\end{align*}
While $\kappa(t)$ is the sum of the semigroup part and the exponential decay, the second term in $\kappa_\alpha(t)$ can be either oscillating or decaying, depending on the value of $x_\alpha$. The oscillations occur if and only if
\begin{equation}
\frac{x_\alpha}{(1-x_\alpha)^2}>\left(\frac{Z}{2\omega}\right)^2.
\end{equation}
This inequality holds e.g. for the evolution from Example \ref{ExZo}, where $x_\alpha=1/d$ and $Z=\frac{\omega}{\pi}\ln(d-1)$, but only in dimensions $d\leq 10$. For $d\geq 11$, oscillating dynamical maps are produced by decaying kernels.
\end{Example}

\begin{Example}
For $\ell(t)=r(d+1)e^{-rt}/d$, one finds
\begin{equation}
\kappa_\alpha(t)=-\frac{r}{d}(d+1)(1-x_\alpha)\left[\delta(t)+S_\alpha e^{S_\alpha t}\right]
\end{equation}
with $S_\alpha=(r/d)[1-(d+1)x_\alpha]$. Note that $\kappa_\alpha(t\to\infty)\to-\infty$ whenever $x_\alpha<1/(d+1)$.
The choice $x_\alpha=1/(d+1)$ reproduces the memory kernel for the Markovian semigroup,
\begin{equation}
\kappa_\alpha(t)=-\frac{r}{d}\delta(t).
\end{equation}
\end{Example}

\section{Conclusions}

We analyzed the convex combinations of non-invertible generalized Pauli dynamical maps. We showed how to choose the mixing parameters to generate additional singularities, shift the existing ones, or remove them altogether. In particular, we demonstrated a method to construct the Markovian semigroup. Next, we analyzed the behavior of time-local generators and memory kernels upon mixing the corresponding dynamical maps. It turned out that every convex combination of the generalized Pauli dynamical maps is provided by a memory kernel defined with a probability distribution and a single function of time. Interestingly, the maps whose mixture is the Markovian semigroup are the solutions of the master equations with the time-local generators and memory kernels that are not regular.

Mixtures of non-invertible dynamical maps have not been thoroughly studied yet. Therefore, there are many open questions regarding this subject. It would be interesting to fully classify the types of quantum evolution that can be obtained by taking convex combinations of dynamical maps. One could ask what happens if maps with different singular points are being mixed. In particular, a mixture of an invertible and a non-invertible dynamical map could be considered. Also, there is an open problem considering Markovianity and non-Markovianity of non-invertible maps. It is unknown how the divisibility of dynamical maps transfers to the properties of their mixture.

\bibliography{C:/Users/cynda/OneDrive/Fizyka/bibliography}
\bibliographystyle{C:/Users/cynda/OneDrive/Fizyka/beztytulow2}

\end{document}